\documentclass[11pt]{amsart}

\usepackage{amsfonts,latexsym}
\usepackage{amsmath}
\usepackage{amstext}
\usepackage{amssymb}
\usepackage{color}
\usepackage{graphicx,diagrams}
\graphicspath{{figures/}}

\newcommand{\R}{\mathbb{R}}

\newcommand{\N}{{\mathbb{N}}}

\newcommand{\argmin}{\textnormal{arg}\min}

\newcommand{\St}{\mathcal{S}}
\newarrow{Tot} ----{triangle}

\newtheorem{theorem}{Theorem}[section]
\newtheorem{lemma}[theorem]{Lemma}

\theoremstyle{definition}
\newtheorem{definition}[theorem]{Definition}

\theoremstyle{remark}
\newtheorem{remark}[theorem]{Remark}

\numberwithin{equation}{section}

\begin{document}

\title[Controller Synthesis for Safety and Reachability via Approximate Bisimulation]
      {Controller Synthesis for Safety and Reachability \\ via Approximate Bisimulation}
\thanks{This work was supported by the Agence Nationale de la Recherche (VEDECY project - ANR 2009 SEGI 015 01).}

%%%%%%%%%%%%%%%%%%%%%%%%%%%%%%%%%%%%%%%%%%%%%%%%%%%%%%%%%%%%%%%

\author[Antoine Girard]{Antoine Girard}
\address{Laboratoire Jean Kuntzmann \\
Universit\'e de Grenoble \\
B.P. 53, 38041 Grenoble, France} \email{antoine.girard@imag.fr}

%%%%%%%%%%%%%%%%%%%%%%%%%%%%%%%%%%%%%%%%%%%%%%%%%%%%%%%%%%%%%%%%

\maketitle

%%%%%%%%%%%%%%%%%%%%%%%%%%%%%%%%%%%%%%%%%%%%%%%%%%%%%%%%%%%%%%%%

\begin{abstract}
In this paper, we consider the problem of controller design using approximately bisimilar abstractions with an emphasis on safety and reachability specifications.
%Safety problems consist in synthesizing a controller that restricts the behaviors of a system so that its output remains in some specified safe set.
%Reachability problems consist in synthesizing a controller that steers the output of the system to some target region while keeping it in a given safe set along the way. 
We propose abstraction-based approaches to solve both classes of problems. We start by synthesizing a controller for an approximately bisimilar abstraction. Then, using a concretization procedure, we obtain a controller for our initial system that is proved ``correct by design''.
We provide guarantees of performance by giving estimates of the distance of the synthesized controller to the maximal (i.e the most permissive) safety controller or
to the time-optimal reachability controller.   
Finally, we use the presented techniques combined with discrete approximately bisimilar abstractions of switched systems developed recently, for switching controller synthesis.
\end{abstract}
\section{Introduction}

The use of discrete abstractions has become a standard approach to hybrid systems design~\cite{raisch1998,moor1999,HabetsS01,tabuada2006,KloetzerB06,reissig2009}.
The benefit of this approach is double. Firstly, by abstracting the
continuous dynamics, controller synthesis problems can be efficiently solved using techniques developed in the areas of supervisory control of discrete-event systems~\cite{Ramadge87} or algorithmic game theory~\cite{arnold2003}. Secondly, if the behaviors of the original system and of the discrete abstraction are formally related by an inclusion or equivalence relationship, the synthesized controller is known to be correct by design and thus the need of formal verification is reduced. 
%This approach enables to synthesize controllers using techniques developped in the areas of supervisory control of discrete-event systems~\cite{ramadge1987} or algorithmic game theory~\cite{arnold2003}. While a lot of this work addressed problems with safety, reachability or some other temporal logic specifications, the use of discrete abstractions for solving optimal control problems~\cite{broucke2000,grune2007} has been less considered.
%
Abstraction, using traditional systems behavioral relationships, relies on inclusion or equality of observed behaviors.
One of the most common notions is that of bisimulation equivalence~\cite{milner1989}. 
However, for systems observed over metric spaces, requiring strict equality
of observed behaviors is often too strong. 
Indeed, the class of continuous or hybrid systems admitting bisimilar discrete abstractions is quite restricted~\cite{alur2000,tabuada2009}. 
In~\cite{A-GirPap07}, a notion of approximate bisimulation, which only asks for closeness of observed behaviors, was introduced.
This relaxation made it possible to extend the class of systems for which discrete abstractions can be 
computed~\cite{A-PolGirTab08,B-GirPolTab08}.

This paper deals with the synthesis of controllers using approximately bisimilar abstractions with an emphasis on safety and reachability problems.
Safety problems consist in synthesizing a controller that restricts the behaviors of a system so that its outputs remain in some specified safe set.
One is usually interested in designing a controller that is as permissive as possible since this makes it possible, using modular approaches, to ensure, a posteriori, secondary control objectives (see e.g.~\cite{Ramadge87}). 
Reachability problems consist in synthesizing a controller that steers the observations of the system to some target region while keeping them in a given safe set along the way. In addition, in order to choose among the possible controllers, we try to minimize the time to reach the target. Hence, we consider a time-optimal control problem.
We propose abstraction-based approaches to solve both classes of problems. We start by synthesizing a controller for an approximately bisimilar abstraction of our concrete system.
Then, using a concretization procedure that is problem-specific, we obtain a controller for our concrete system that is proved ``correct by design''.
For safety problems, we provide estimates of the distance between the synthesized controller and the maximal (i.e the most permissive) safety controller.
For reachability problems, we provide estimates of the distance between the performances of the synthesized controller and of the time-optimal controller.  
As an illustration, we use these techniques in combination with the discrete approximately bisimilar abstractions of switched systems developed in~\cite{B-GirPolTab08}, for switching controller synthesis.
%The main practical contributions of the paper are the concretization procedures for safety and reachability controllers.
%Compared to the natural controller {\it concretization} procedure, described in~\cite{tabuada2009,girard10} and used in~\cite{A-PolTab09,A-PolGirTab08,B-GirPolTab08}, which essentially renders  two controlled systems approximately bisimilar, the controller we obtain have several advantages.
 %This is the approach used in~\cite{A-PolTab09,A-PolGirTab08,B-GirPolTab08}
Preliminary versions of these results appeared in the conference papers~\cite{girard10,girard10a}. The presentation has been improved and new results
on estimates of the distance to maximal or optimal controllers have been added. Deeper numerical experiments have also been provided.

Controller synthesis using approximately (bi)similar abstractions has also been considered in~\cite{tazaki08,mazo2010}.
In~\cite{tazaki08}, the authors use approximately bisimilar abstractions to design a suboptimal controller for a fixed bounded horizon optimal control problem.
In this paper, we consider time-optimal control for reachability specifications, thus the time horizon is variable. Our work is more closely related to~\cite{mazo2010} where time-optimal control
is considered as well. We shall discuss further in the paper the main differences with our approach.
Regarding safety specifications,
this is the first paper proposing a specific approach using approximately bisimilar abstractions.

%The paper is organized as follows. In Section 2, we briefly introduce the modeling framework of transition systems and the notion of approximate bisimulation relation.
%Sections 3 and 4 deal with abstraction-based synthesis of safety and reachability controllers, respectively. 
%Finally, in Section 5, as an illustration, we use these techniques in combination with the discrete approximately bisimilar abstractions of switched systems developed in~\cite{B-GirPolTab08}, for switching controller synthesis.

\section{Preliminaries}
\label{sec:prel}
%In this section, we define the classes of systems and controllers that we consider in this paper.
%We also present the notion of approximate bisimulation~\cite{A-GirPap07} which yields approximation relations between systems. 

%\subsection{Systems and Controllers}

We start by introducing the class of transition systems which serves as a common abstract modeling framework for
discrete, continuous or hybrid systems (see e.g.~\cite{alur2000,tabuada2009}).

\begin{definition}
A \textit{transition system} is a tuple
$T=(Q,L,\delta,O,H)
$
consisting of a set of states $Q$;
a set of actions $L$; a transition relation $\delta \subseteq Q\times L \times Q$;
a set of observations $O$;
an output function $H:Q\rightarrow O$.
$T$ is said to be 
\textit{discrete} if $Q$ and $L$ are finite or countable sets, 
\textit{metric} if the set of observations $O$ is equipped with a metric ${d}$.
\end{definition}

The transition $(q,l,q')\in \delta$ will be denoted $q'\in \delta(q,l)$; this 
means that the system can evolve from state $q$ to state $q'$ under the action $l$.
Given a subset of actions $L'\subseteq L$, we denote $\delta(q,L')=\bigcup_{l\in L'} \delta(q,l)$.
An action $l\in L$ belongs to the set of {\it enabled actions} at state $q$, denoted $\textrm{Enab}(q)$, if 
$\delta(q,l) \ne \emptyset$.
If $\textrm{Enab}(q)=\emptyset$, then $q$ is said to be a {\it blocking} state; otherwise it is said to be {\it non-blocking}.
If all states are non-blocking, we say that the transition system $T$ is non-blocking.
The transition system is said to be {\it non-deterministic} if there exists $q\in Q$ and $l\in \textrm{Enab}(q)$ such that $\delta(q,l)$ has several elements.
A {\it trajectory} of the transition system is a finite sequence of states and actions
$
(q_0,l_0),(q_1,l_1),\dots,(q_{N-1},l_{N-1}),q_N
$  
where
$q_{i+1}\in \delta(q_i,l_i)$ for all $i\in \{0,\dots,N-1\}$. $N\in \N$ is referred to as the {\it length} of the trajectory.
%A trajectory of length $0$ consists of a single state $q_0$.
The {\it observed behavior} associated to trajectory is the finite sequence of observations $o_0 o_1 o_2 \dots o_N$ where 
$o_i=H(q_i)$, for all $i\in \{0,\dots, N\}$.

This paper deals with controller synthesis for transition systems; we shall consider only static (i.e. without memory) state-feedback controllers. However, we will just use the term controller for brevity.
\begin{definition}
\label{def:controller}
A \textit{controller} for transition system $T$ is a map 
$\St:Q \rightarrow 2^L.$
It is {\it well-defined} if $\St(q) \subseteq \textnormal{Enab}(q)$, for all  $q\in Q$.
%\begin{equation}
%\label{eq:stat0}
%\forall q\in Q,\; \St(q) \subseteq \textnormal{Enab}(q).
%\end{equation}
The dynamics of the controlled system is described by the transition system $T_\St=(Q,L,\delta_\St,O,H)$ where the transition
relation is given by
$q'\in \delta_\St(q,l)$ if and only if $l \in \St(q)$ and $q'\in \delta(q,l)$.
\end{definition}

Given a subset of states $Q'\subseteq Q$, we will denote $\St(Q')=\bigcup_{q\in Q'} \St(q)$.
Let us remark that a state $q$ of $T_\St$ is non-blocking if and only if $\St(q)\ne \emptyset$.
A controller essentially executes as follows. The state $q$ of $T$ is measured, an action 
$l\in \St(q)$ is selected and actuated. Then, the system takes a transition $q'\in \delta(q,l)$, this is always possible if $\St$ is well-defined.

%\subsection{Approximate Bisimulation}

In this paper, we consider approximate equivalence relationships for transition systems 
defined by approximate bisimulation relations introduced in~\cite{A-GirPap07}.

%rely on equality of observed behaviors.
%One of the most common notion is that of bisimulation equivalence~\cite{milner1989,tabuada2009}.
%A bisimulation relation between two transition systems $T_1$ and $T_2$
%is a relation between their set of states describing how to select transitions of $T_1$ to match transitions of $T_2$
%(and conversely) in order to produce identical observed behaviors.
%Requiring strict equivalence of behaviors, for transition systems observed over metric spaces is often too restrictive.
%A natural relaxation is to ask for closeness of observed behaviors where closeness is measured with respect to the
%metric on the observation set. This leads to the notion of approximate bisimulation introduced in~\cite{A-GirPap07}.

\begin{definition}\label{def:bisim} Let $T_i=(Q_i,L,\delta_i,O,H_i)$, $i=1,2$,
 be two metric transition systems with the
same sets of actions $L$ and observations $O$ equipped with the 
metric $d$, let $\varepsilon \ge 0$ be a given precision. 
A
relation $R \subseteq Q_1\times Q_2$ is said to be an
{\it $\varepsilon$-approximate bisimulation relation} between $T_1$ and $T_2$
if, for~all $(q_1,q_2)\in R$:
\begin{itemize}
\item $d(H_1(q_1),H_2(q_2)) \le \varepsilon$;

\item $\forall l\in \textnormal{Enab}_1(q_1)$, $\forall q_1'\in \delta_1(q_1,l)$, $\exists q_2' \in \delta_2(q_2,l)$, such that $(q_1',q_2')\in R$;

\item $\forall l\in \textnormal{Enab}_2(q_2)$, $\forall q_2'\in \delta_2(q_2,l)$, $\exists q_1' \in \delta_1(q_1,l)$, such that $(q_1',q_2')\in R$.

\end{itemize}
The transition systems $T_1$ and $T_2$ are said to be
{\it approximately bisimilar} with precision $\varepsilon$, denoted
$T_1\sim_\varepsilon T_2$, if:
\begin{itemize}
\item $\forall q_1\in Q_1$, $\exists q_2 \in Q_2$, such that
$(q_1,q_2) \in R$;

\item $\forall q_2\in Q_2$, $\exists q_1 \in Q_1$, such that
$(q_1,q_2) \in R$.
\end{itemize}
\end{definition}
If $T_1$ is a system we want to control and $T_2$ is a simpler system that we want to use for controller synthesis, then $T_2$ is called an {\it approximately bisimilar abstraction} of $T_1$. 

We will denote
for $q_1\in Q_1$, $R(q_1)=\{q_2\in Q_2|\; (q_1,q_2)\in R\}$ and for $Q_1'\subseteq Q_1$, $R(Q'_1)=\bigcup_{q_1 \in Q_1'} R(q_1)$; for $q_2\in Q_2$, $R^{-1}(q_2)=\{q_1\in Q_1|\; (q_1,q_2)\in R\}$ and
for $Q_2'\subseteq Q_2$, $R^{-1}(Q'_2)=\bigcup_{q_2 \in Q_2'} R^{-1}(q_2)$.
%For $\varepsilon=0$, we recover the usual notion of exact bisimulation equivalence~\cite{milner1989,tabuada2009}. 
%It is important to note that, in approximate bisimulation relations, one state $q_1$ of $T_1$ can be related to several states of $T_2$ ($R(q_1)$ has several elements), contrarily to exact bisimulation relations where $Q_2$ is generally given by a partition of the set of states $Q_1$ and therefore one state $q_1$ of $T_1$ is related to at most one state of $T_2$ ($R(q_1)$ is a singleton). 

\begin{remark} We assume that systems $T_1$ and $T_2$ have the same sets of actions and that matching transitions in the second and third items of the previous definition share the same input. These conditions can actually be relaxed using the notion of alternating approximate bisimulation relation~\cite{tabuada2009}. The results presented in this paper can be easily extended to that setting.
\end{remark}

The problem of computing approximately bisimilar discrete abstractions has been considered 
for nonlinear control systems~\cite{A-PolGirTab08} and switched systems~\cite{B-GirPolTab08}.
A controller designed for an abstraction can be used to synthesize a controller for the concrete system via a natural {\it concretization} procedure, described in~\cite{tabuada2009,girard10}, which essentially renders the two controlled systems
approximately bisimilar. This is the approach used in~\cite{A-PolGirTab08,B-GirPolTab08,
mazo2010}. However, the controller for the concrete system obtained via this concretization procedure has several drawbacks. 
Firstly, it is generally a dynamic state-feedback controller (i.e. the controller has a memory) when it is known that for some control specifications such as safety~\cite{Ramadge87} or reachability~\cite{bertsekas2000}, it is sufficient to consider static state-feedback controllers.
Secondly, the implementation of this controller requires the encoding of the dynamics of the abstraction which may result in a higher implementation cost. Thirdly, 
if the abstraction is deterministic, then there is essentially no more feedback: at each step the controller selects the control action and its internal state independently from the state of the concrete system. This may cause some robustness issues in case of unmodeled disturbances or fault occurrences. 
In the following, we present specific controller concretization procedure for safety and reachability specifications which do not suffer from the previous drawbacks. These techniques are readily applicable using the discrete abstractions mentioned above.

\section{Controller Synthesis for Safety Specifications}

%The main purpose of this section is the design of controllers for safety specifications; the goal is to keep the observations of the system in some specified safe set. 
%Given a system $T_1$ that we want to control, we want to use an approximately bisimilar abstraction $T_2$ to synthesize an abstract controller which after concretization can serve for $T_1$.

\subsection{Problem Formulation}

Let $T=(Q,L,\delta,O,H)$ be a transition system,
let $O_s \subseteq O$ be a set of outputs associated with safe states. 
We consider the synthesis problem that consists in determining a controller that keeps the output of the system inside the specified safe set $O_s$. 
\begin{definition}  
A controller $\St$ for $T$ is a {\it safety controller for specification $O_s$} if, 
for all non-blocking states $q_0$
of the controlled system $T_\St$ (i.e. $\St(q_0)\ne \emptyset$),
for all trajectories of $T_\St$ starting from $q_0$, $(q_0,l_0),(q_1,l_1),\dots,(q_{N-1},l_{N-1}),q_N$;
for all $i\in \{0,\dots, N\}$,  $H(q_i) \in O_s$ and $q_N$ is a non-blocking state of $T_\St$ (i.e. $\St(q_N)\ne \emptyset$).
%
%$$q_0\rTot_\St^{l_0} q_1 \rTot_\St^{l_1}  \dots \rTot_\St^{l_{N-1}} q_N,$$  
%the following conditions hold:
%\begin{itemize}
%\item For all $i\in \{0,\dots, N\}$,  $H(q_i) \in O_s$;
%
%\item $q_N$ is a non-blocking state of $T_\St$ (i.e. $\St(q_N)\ne \emptyset$).
%\end{itemize}
%
\end{definition}
The condition that all trajectories end in a non-blocking state ensures that, starting from a non-blocking state, the controlled system can evolve indefinitely while keeping its output in the safe set $O_s$. It is easy to verify by induction that an equivalent characterization of safety controllers is given as follows:
\begin{lemma}
\label{lem:safety}
A controller $\St$ for $T$ is a safety controller for specification $O_s$ if and only if for all non-blocking states $q$ 
of the controlled system $T_\St$ (i.e. $\St(q)\ne \emptyset$);
$H(q) \in O_s$ and
for all $q'\in \delta(q,\St(q))$, $q'$ is a non-blocking state of $T_\St$ (i.e. $\St(q')\ne \emptyset$).
%, the following conditions hold:
%\begin{itemize}
%\item $H(q) \in O_s$;
%\item For all $q'\in \delta(q,\St(q))$, $q'$ is a non-blocking state of $T_\St$ (i.e. $\St(q')\ne \emptyset$).
%\end{itemize}
\end{lemma}

There are in general several controllers that solve the safety problem.
We are usually interested in synthesizing a controller that enables as many actions as possible.
%Indeed, this makes it possible, using modular approaches, to ensure, a posteriori, secondary control objectives (see e.g.~\cite{Ramadge87}).
This notion of permissivity can be formalized by defining a partial order on controllers.
\begin{definition}\label{def:max} Let $\St_1$ and $\St_2$ be two controllers for transition system $T$,
$\St_1$ is {\it more permissive} than $\St_2$, denoted $\St_2 \preceq \St_1$, if for all $q\in Q$, $\St_2(q) \subseteq \St_1(q)$. The controller $\St^*$ for $T$ is the {\it maximal safety controller for specification $O_s$}, if 
$\St^*$ is a safety controller for specification $O_s$, and for all safety controllers $\St$ for specification $O_s$, $\St\preceq \St^*$.
\end{definition}

%Essentially, $\St_2 \preceq \St_1$ means that each time an action is enabled by $\St_2$, it is also enabled by $\St_1$. It follows that all
%trajectories of $T_{\St_2}$ are also trajectories of $T_{\St_1}$.
%Regarding our synthesis problem, i
It is well known that the maximal safety controller exists, is unique
and can be computed using a fixed point algorithm (see e.g.~\cite{Ramadge87,tabuada2009}).
This algorithm is guaranteed to terminate in a finite number of steps provided $H^{-1}(O_s)\subseteq Q$ is a finite set which is often the case
for discrete systems. For other transition systems, there is no guarantee that the algorithm will terminate.
In this case, a synthesis approach based on approximately bisimilar abstractions can help to compute effectively a safety controller with, in addition, an estimation of the distance to maximality.

\subsection{Abstraction-based Controller Synthesis}
\label{sec:safe}

Let $T_i=(Q_i,L,\delta_i,O,H_i)$, $i=1,2$, be metric transition systems such that $T_1\sim_\varepsilon T_2$. 
Let $T_1$ be the system that we want to control and $T_2$ be an approximately bisimilar abstraction of $T_1$.
We present an approach to safety controller synthesis for specification $O_s$.
%
%The first step of our approach consists in computing a safety controller for $T_2$ and a modified specification where the safe set is given by a contraction of $O_s$.
\begin{definition} Let $O' \subseteq O$ and $\varphi \ge 0$. The {\it $\varphi$-contraction} of $O'$ is the subset of $O$ defined as follows
$$
C_\varphi(O')=\left \{ o' \in O' |\; \forall o \in O,\; d(o,o') \le \varphi \implies o \in O' \right\}.
$$
The {\it $\varphi$-expansion} of $O'$  is the subset of $O$ defined as follows
$$
E_\varphi(O')=\left \{ o \in O |\; \exists o' \in O', d(o,o') \le \varphi \right\}.
$$
\end{definition}

By straightforward applications of the previous definitions, we have:
\begin{lemma}\label{lem:contract} Let $O'\subseteq O$ and $\varepsilon\ge 0$, then 
$C_{2\varepsilon}(O') \subseteq C_{\varepsilon}(C_\varepsilon(O'))$,
$E_{\varepsilon}(E_\varepsilon(O')) \subseteq E_{2\varepsilon}(O')$ and
$O'  \subseteq C_{\varepsilon}(E_\varepsilon(O'))$.
\end{lemma}

%\begin{proof} For the first inclusion, let $o_1\in C_{2\varepsilon}(O')$, let $o_2 \in O$ such that $d(o_1,o_2)\le \varepsilon$. We need to prove that $o_2\in C_\varepsilon(O')$. Let $o_3\in O$
%such that $d(o_2,o_3)\le \varepsilon$, then by the triangular inequality we have $d(o_1,o_3) \le d(o_1,o_2)+d(o_2,o_3) \le 2 \varepsilon$. 
%Since $o_1\in C_{2\varepsilon}(O')$, we have
%that $o_3 \in O'$. Therefore, $o_2 \in  C_\varepsilon(O')$ and it follows that $o_1 \in  C_{\varepsilon}(C_\varepsilon(O'))$. Therefore, $C_{2\varepsilon}(O') \subseteq C_{\varepsilon}(C_\varepsilon(O'))$.
%
%For the second inclusion, let $o_1 \in E_{\varepsilon}(E_\varepsilon(O'))$. There exists $o_2\in E_\varepsilon(O')$ such that $d(o_1,o_2)\le \varepsilon$. There also exists $o_3 \in O'$
%such that $d(o_2,o_3)\le \varepsilon$. Then  by the triangular inequality we have $d(o_1,o_3) \le d(o_1,o_2)+d(o_2,o_3) \le 2 \varepsilon$ which gives $o_1 \in 
%E_{2\varepsilon}(O')$. Therefore, $E_{\varepsilon}(E_\varepsilon(O')) \subseteq E_{2\varepsilon}(O')$.
%
%For the third inclusion, let $o_1 \in O'$, let $o_2 \in  O$ such that $d(o_1,o_2)\le \varepsilon$. We need to prove that $o_2 \in E_\varepsilon(O')$ which is obvious since $o_1\in O'$ and $d(o_1,o_2)\le \varepsilon$. Therefore, 
%$O'  \subseteq C_{\varepsilon}(E_\varepsilon(O'))$.
%\end{proof}

%\medskip
We start by synthesizing a safety controller for the abstraction $T_2$ and the specification $C_\varepsilon(O_s)$.
This controller is denoted $\St_{2,C_\varepsilon}$. We shall not discuss further the synthesis of this controller which can be done, if $T_2$ is discrete, using a fixed point algorithm.
The second step of our approach allows us to design a safety controller for system $T_1$ and specification $O_s$, obtained from the controller $\St_{2,C_\varepsilon}$ using the following concretization procedure:
\begin{theorem}
\label{th:stat}
Let $T_1 \sim_\varepsilon T_2$, let  $R\subseteq Q_1 \times Q_2$
denote the $\varepsilon$-approximate bisimulation relation between $T_1$ and $T_2$.
Let $\St_{2,C_\varepsilon}$ 
be a safety controller for $T_2$ and specification $C_\varepsilon(O_s)$.
Let us define $\St_{1}$, the controller for $T_1$
given by 
\begin{equation}
\label{eq:ref}
\forall q_1\in Q_1,\; \St_{1}(q_1) =  \St_{2,C_\varepsilon}(R(q_1)).
\end{equation}
Then, $\St_{1}$ is well-defined and is a safety controller for specification $O_s$.
\end{theorem}

\begin{proof} 
First, let us show that the controller $\St_{1}$ is well-defined. Let $q_1\in Q_1$, let $l\in \St_1(q_1)$ then, from~(\ref{eq:ref}), there exists $q_2\in Q_2$ such that $(q_1,q_2)\in R$ and $l\in \St_{2,C_\varepsilon}(q_2)$. $\St_{2,C_\varepsilon}$ is well-defined, then
$l\in \text{Enab}_2(q_2)$, i.e. there exists $q_2'\in \delta_2(q_2,l)$.
By Definition~\ref{def:bisim}, it follows that there exists $q'_1\in \delta_1(q_1,l)$ (such that $(q_1',q_2')\in R$), which implies that 
$l\in \text{Enab}_1(q_1)$. Thus, for all $q_1\in Q_1$, $\St_1(q_1)\subseteq \text{Enab}_1(q_1)$;
$\St_{1}$ is well-defined.
Let us now prove that $\St_1$ is a safety controller for the specification $O_s$. 
Let $q_{1}\in Q_1$ such that $\St_1(q_1)\ne \emptyset$, 
let $l\in \St_1(q_1)$, by (\ref{eq:ref}) there
exists $q_2\in Q_2$ such that $(q_1,q_2)\in R$ and $l\in \St_{2,C_\varepsilon}(q_2)$. 
Since $\St_{2,C_\varepsilon}$ is a safety controller for specification $C_\varepsilon(O_s)$ and $\St_{2,C_\varepsilon}(q_2)\ne \emptyset$, we have 
from Lemma~\ref{lem:safety}, that $H_2(q_2)\in C_\varepsilon(O_s)$. By Definition~\ref{def:bisim}, $d(H_1(q_1),H_2(q_2)) \le \varepsilon$ and therefore $H_1(q_1)\in O_s$.
Now, let $q_1'\in\delta(q_1,l)$, by Definition~\ref{def:bisim} there exists $q_2'\in \delta_2(q_2,l)$ such that $(q_1',q_2')\in R$.
Since $\St_{2,C_\varepsilon}$ is a safety controller for specification $C_\varepsilon(O_s)$ and $l\in \St_{2,C_\varepsilon}(q_2)$, we have 
from Lemma~\ref{lem:safety}, that  $\St_{2,C_\varepsilon}(q_2')\ne \emptyset$.
Finally, (\ref{eq:ref}) implies that $\St_{2,C_\varepsilon}(q_2') \subseteq \St_1(q_1')$ and therefore $\St_1(q_1')\ne \emptyset$.
From Lemma~\ref{lem:safety}, $\St_{1}$ is a safety controller for specification $O_s$.
\end{proof}

%The previous theorem gives us a way to concretize a safety controller for abstraction $T_2$ into a safety controller for $T_1$. 
If we use the maximal safety controller  for $T_2$, 
it is desirable to have an estimate of the distance between the controller given by the concretization equation (\ref{eq:ref}) and the maximal safety controller for $T_1$.
This is given by the following result:
\begin{theorem}
\label{th:approxsafe1} 
Let $\St_{2,C_\varepsilon}^*$ and $\St_{2,E_\varepsilon}^*$ be the maximal safety controllers for $T_2$ and specifications $C_\varepsilon(O_s)$ and $E_\varepsilon(O_s)$ respectively.
Let $\St_{1}$ and $\St_{1,E_{2\varepsilon}}$ be the controllers for $T_1$ obtained by the concretization equation (\ref{eq:ref}) from $\St_{2,C_\varepsilon}^*$ and $\St_{2,E_{\varepsilon}}^*$ respectively.
Let $\St_{1}^*$, $\St_{1,C_{2\varepsilon}}^*$ and $\St_{1,E_{2\varepsilon}}^*$ be the maximal safety controllers for $T_1$ and specifications $O_s$, 
$C_{2\varepsilon}(O_s)$ and $E_{2\varepsilon}(O_s)$ respectively.
 Then,
$$
\St_{1,C_{2\varepsilon}}^* \preceq \St_{1} \preceq \St_{1}^* \preceq \St_{1,E_{2\varepsilon}} \preceq \St_{1,E_{2\varepsilon}}^*.
$$ 
\end{theorem}

\begin{proof} The proof relies on the introduction of several auxiliary controllers. The relations between these controllers are presented in the following sketch, where arrows correspond to application of the concretization equation~(\ref{eq:ref}):
$$
\begin{diagram}
\St^*_{1,C_{2\varepsilon}} & \preceq & \tilde \St_1 & \preceq  & \St_1 & \preceq & \St_1^*
& \preceq & \St_{1,E_{2\varepsilon}} & \preceq & \St^*_{1,E_{2\varepsilon}}\\
\\
&\rdTo~{\text{\footnotesize (\ref{eq:ref})}}  & \uTo~{\text{\footnotesize (\ref{eq:ref})}}  & &\uTo~{\text{\footnotesize (\ref{eq:ref})}} & & \dTo~{\text{\footnotesize (\ref{eq:ref})}} &  &\uTo~{\text{\footnotesize (\ref{eq:ref})}}\\
 & & \St_{2,C_\varepsilon} & \preceq & \St^*_{2,C_\varepsilon} & &\St_{2,E_\varepsilon} & \preceq & \St^*_{2,E_\varepsilon}
\end{diagram}
$$
Let us go into the details of the proof.
From Theorem~\ref{th:stat}, $\St_1$ is a safety controller for specification $O_s$, then 
%by definition of maximal safety controller,
$\St_{1} \preceq \St_{1}^*$.
Let us prove that $\St_{1,C_{2\varepsilon}}^* \preceq \St_{1}$. Since $\St_{1,C_{2\varepsilon}}^*$ is a safety controller for specification $C_{2\varepsilon}(O_s)$ and from Lemma~\ref{lem:contract}, $C_{2\varepsilon}(O_s) \subseteq C_{\varepsilon}(C_\varepsilon(O_s))$, it is clear that
$\St_{1,C_{2\varepsilon}}^*$ is a safety controller for specification $C_{\varepsilon}(C_\varepsilon(O_s))$.
Now, let us define $\St_{2,C_\varepsilon}$, the controller for $T_2$
such that for $q_2\in Q_2$, $\St_{2,C_\varepsilon}(q_2) = {\St}_{1,C_{2\varepsilon}}^*(R^{-1}(q_2))$. 
The symmetry of approximate bisimulation allows us to reverse the role of $T_1$ and $T_2$ in Theorem~\ref{th:stat}.
This gives that $\St_{2,C_\varepsilon}$
is a safety controller for $T_2$ and specification $C_\varepsilon(O_s)$ which yields 
%Let $\St_{2,C_\varepsilon}^*$ 
%be the maximal safety controller for $T_2$ and specification $C_\varepsilon(O_s)$. Then, 
$\St_{2,C_\varepsilon} \preceq \St_{2,C_\varepsilon}^*$.
We now define $\tilde\St_{1}$, the controller for $T_1$ such that for $q_1\in Q_1$,
$\tilde \St_{1}(q_1) = \St_{2,C_\varepsilon}(R(q_1))$.
Then,  $\St_{2,C_\varepsilon} \preceq \St_{2,C_\varepsilon}^*$ gives $\tilde\St_{1} \preceq \St_1$.
Finally, we remark that for all $q_1\in Q_1$
$\tilde \St_{1}(q_1) = {\St}_{1,C_{2\varepsilon}}^*(R^{-1}(R(q_1))))
$
which leads to ${\St}_{1,C_{2\varepsilon}}^* \preceq \tilde{\St}_{1}$.
Let us show that $\St_{1}^*  \preceq  \St_{1,E_{2\varepsilon}}$.
Since $\St_1^*$ is a safety controller for $O_s$ and since from Lemma~\ref{lem:contract}, $O_s \subseteq C_{\varepsilon}(E_\varepsilon(O_s))$, it is clear that
$\St_{1}^*$ is a safety controller for specification $C_{\varepsilon}(E_\varepsilon(O_s))$. Now let us define $\St_{2,E_\varepsilon}$, the controller for $T_2$
such that for $q_2\in Q_2$, $\St_{2,E_\varepsilon}(q_2) ={\St}_{1}^*(R^{-1}(q_2))$.
By reversing the role of $T_1$ and $T_2$ in Theorem~\ref{th:stat},
we obtain that $\St_{2,E_\varepsilon}$
is a safety controller for $T_2$ and specification $E_\varepsilon(O_s)$. Then,  
$
\St_{2,E_\varepsilon} \preceq \St_{2,E_\varepsilon}^*$. Then, for all $q_1\in Q_1$,
$
\St_1^*(q_1) \subseteq  {\St}_{1}^*(R^{-1}(R(q_1))) = 
{\St}_{2,E_\varepsilon}(R(q_1)) \subseteq {\St}_{2,E_\varepsilon}^*(R(q_1))
=
\St_{1,E_{2\varepsilon}}(q_1).
$
Hence,  $\St_{1}^*  \preceq  \St_{1,E_{2\varepsilon}}$.
Finally,
since ${\St}_{2,E_\varepsilon}^*$ is a safety controller for $T_2$ and specification $E_\varepsilon(O_s)$ and since by Lemma~\ref{lem:contract},
$E_\varepsilon(O_s) \subseteq  C_\varepsilon(E_{\varepsilon}(E_{\varepsilon}(O_s))) \subseteq C_\varepsilon(E_{2\varepsilon}(O_s))$, it follows
that ${\St}_{2,E_\varepsilon}^*$ is a safety controller for $T_2$ and specification $ C_\varepsilon(E_{2\varepsilon}(O_s))$. Then, from Theorem~\ref{th:stat}, $\St_{1,E_{2\varepsilon}}$ is a safety controller for $T_1$ and specification $E_{2\varepsilon}(O_s)$ which yields $\St_{1,E_{2\varepsilon}} \preceq \St_{1,E_{2\varepsilon}}^*$.
\end{proof}

By computing the controllers $\St_1$ and $\St_{1,E_{2\varepsilon}}$ one is able to give a certified upper-bound on the distance between the controller $\St_1$ we will 
use to control $T_1$ and the maximal safety controller $\St_1^*$.
Moreover, if the safety problem is somehow robust, in the sense that $\St^*_{1,C_{2\varepsilon}}$ and
$\St^*_{1,E_{2\varepsilon}}$ approach $\St_1^*$ as $\varepsilon$ approaches $0$ (i.e. slightly different specifications
result in only slightly different maximal controllers); then 
$\St_1$ and $\St_{1,E_{2\varepsilon}}$  also approach
$\St_1^*$ as $\varepsilon$ gets smaller and $\St_1^*$ can be approximated arbitrarily
close.

\section{Controller Synthesis for Reachability Specifications}

%In this section, we establish comparable results for the synthesis of controllers with reachability specifications; 
%the goal is to steer the observations of the system to some target region while keeping them in a given safe set along the way. In addition, in order to choose among the possible controllers, we try to minimize the time to reach the target. Thus, we consider an optimal control problem.
%Here again, given a system $T_1$ that we want to control, we want to use an approximately bisimilar abstraction $T_2$ to synthesize an abstract controller which after concretization can serve for $T_1$.
 
\subsection{Problem Formulation}

Let $T=(Q,L,\delta,O,H)$ be a transition system, let $O_s\subseteq O$ be a set of outputs associated with safe states, let $O_t\subseteq O_s$ be a set of outputs associated with target states. 
We consider the synthesis problem that consists in determining a controller steering the output of the system to $O_t$ while keeping the output in $O_s$ along the way. In addition, in order to choose among the possible controllers, we try to minimize the time to reach the target. Thus, we consider an optimal control problem.
In this section, we assume for simplicity, that $T$ is non-blocking; it would actually be sufficient to assume that all the states of $T$ associated to observations in $O_s$ are non-blocking.

\begin{definition}
Let $\St$ be a controller for $T$ such that for all $q\in Q$, $\St(q) \ne \emptyset$.
The {\it entry time of $T_\St$ from $q_0\in Q$ for reachability specification $(O_s,O_t)$} 
is the smallest $N\in \N$ such that for all trajectories of the controlled system $T_\St$, of length $N$ and starting from $q_0$,
$(q_0,l_0),(q_1,l_1),\dots,(q_{N-1},l_{N-1}),q_N$, there exists $K\in \{0,\dots,N\}$ such that
for all $k \in \{0,\dots,K\},\; H(q_k) \in O_s$ and $ H(q_K)\in O_t$.
%\begin{itemize}
%\item For all $k \in \{0,\dots,K\},\; H(q_k) \in O_s$;
%\item $ H(q_K)\in O_t$.
%\end{itemize}
The entry time is denoted by $J(T_\St,O_s,O_t,q_0)$.
If such a $N\in \N$ does not exist, then we define $J(T_\St,O_s,O_t,q_0)= + \infty$.
\end{definition}

The condition that $\St(q)\ne \emptyset$, for all $q\in Q$, ensures that the controlled system $T_\St$ is non-blocking.
The states from which the system is guaranteed to reach $O_t$ without leaving $O_s$ are the states with finite entry-time.
The following result is quite standard (see e.g.~\cite{bertsekas2000}) and is therefore stated without proof:
\begin{lemma}
\label{lem:reachability}
The {\it entry time} of $T_\St$ for reachability specification $(O_s,O_t)$ satisfies:
\begin{itemize}
\item For all $q\in Q\setminus H^{-1}(O_s), \; J(T_\St,O_s,O_t,q) =+\infty$,
\item For all $q\in H^{-1}(O_t), \; J(T_\St,O_s,O_t,q) =0$,
\item For all $q\in H^{-1}(O_s)\setminus H^{-1}(O_t)$,
\begin{equation}
\label{eq:rec1}
\hspace{-0.7cm}
J(T_\St,O_s,O_t,q) = 1+ \max_{q'\in \delta(q,\St(q))} J(T_\St,O_s,O_t,q'). 
\end{equation}
\end{itemize}
\end{lemma}

We can now define the notion of time-optimal controller:
\begin{definition}
\label{def:opt}
We say that a controller $\St^*$ for $T$ is {\it time-optimal for reachability specification $(O_s,O_t)$}  if for all controllers $\St$, for all $q\in Q$,
$J(T_{\St^*},O_s,O_t,q) \le J(T_\St,O_s,O_t,q)$.
The {\it time-optimal value function for reachability specification $(O_s,O_t)$} is defined as
$J^*(T,O_s,O_t,q)=J(T_{\St^*},O_s,O_t,q).
$
\end{definition}

Solving the time-optimal control problem consists in synthesizing a time-optimal controller.
It is well known that a time-optimal controller exists (but may be not unique) and can be computed using dynamic programming~\cite{bertsekas2000,tabuada2009}.
The dynamic programming algorithm is guaranteed to terminate in a finite number of steps provided $H^{-1}(O_s)\subseteq Q$ is a finite set
which is often the case for discrete systems. Here again, for other systems, there is no guarantee that the algorithm will terminate and an abstraction-based approach
may help to compute a sub-optimal controller with an estimation of the distance to optimality.

\subsection{Abstraction-based Controller Synthesis}
\label{sec:reach}

Let $T_i=(Q_i,L,\delta_i,O,H_i)$, $i=1,2$, be metric transition systems such that $T_1\sim_\varepsilon T_2$. 
Let $T_1$ be the system that we want to control and $T_2$ be an approximately bisimilar abstraction of $T_1$.
We present an approach to controller synthesis for reachability specifications.
We first synthesize a controller $\St_{2,C_\varepsilon}$ for the abstraction $T_2$ and the reachability specification given by the contracted safe set $C_\varepsilon(O_s)$ and target set $C_\varepsilon(O_t)$.  If $T_2$ is discrete, this can be done using dynamic programming. Then, we design a controller for $T_1$ and reachability specification
$(O_s,O_t)$ using the following concretization procedure:

\begin{theorem}
\label{th:stat2}
Let $T_1 \sim_\varepsilon T_2$, let  $R\subseteq Q_1 \times Q_2$
denote the $\varepsilon$-approximate bisimulation relation between $T_1$ and $T_2$.
Let $\St_{2,C_\varepsilon}$ be a controller for $T_2$,
let us define $\St_{1}$, the controller for $T_1$
given by\footnote{If there are several states $q_2\in R(q_1)$ minimizing 
$J(T_{2,\St_{2,C_\varepsilon}},C_{\varepsilon}(O_s),C_{\varepsilon}(O_t),q_2)$ then we just pick one of them.}
\begin{equation}
\label{eq:ref3}
\forall q_1\in Q_1,\; \St_{1}(q_1) = \St_{2,C_\varepsilon}\left(\argmin_{q_2\in R(q_1)} J(T_{2,\St_{2,C_\varepsilon}},C_{\varepsilon}(O_s),C_{\varepsilon}(O_t),q_2)\right)
\end{equation}
where $q_2 \in R(q_1)$ stands for $(q_1,q_2)\in R$. Then, $\St_{1}$ is well-defined and for all $q_1 \in Q_1$:
\begin{equation}
\label{eq:ref4}
J(T_{1,{{\St}_1}},O_s,O_t,q_1) \le \min_{q_2\in R(q_1)} J(T_{2,\St_{2,C_\varepsilon}},C_{\varepsilon}(O_s),C_{\varepsilon}(O_t),q_2).
\end{equation}
\end{theorem}

%\begin{remark}
%Regarding  equation (\ref{eq:ref3}), if there are several states $q_2\in R(q_1)$ minimizing the function
%$J(T_{2,\St_{2,C_\varepsilon}},C_{\varepsilon}(O_s),C_{\varepsilon}(O_t),q_2)$ then any of these states can be used
%when computing $\St_{1}(q_1)$.
%\end{remark}

\begin{proof} The fact that $\St_{1}$ is well-defined can be shown similarly to the proof of Theorem~\ref{th:stat}.
%First, let us show that the controller $\St_{1}$ is well-defined. Let $q_1\in Q_1$, let $l\in \St_1(q_1)$ then, from the concretization equation (\ref{eq:ref3}),
%there exists $q_2\in Q_2$ such that $(q_1,q_2)\in R$ and $l\in \St_{2,C_\varepsilon}(q_2)$. $\St_{2,C_\varepsilon}$ is well-defined, then
%$l\in \text{Enab}_2(q_2)$, i.e. there exists $q'_2 \in \delta(q_2,l)$.
%By Definition~\ref{def:bisim}, it follows that there exists  $q'_1\in \delta_1(q_1,l)$ (such that $(q_1',q_2')\in R$), which implies that 
%$l\in \text{Enab}_1(q_1)$. Thus, for all $q_1\in Q_1$, $\St_1(q_1)\subseteq \text{Enab}_1(q_1)$;
%$\St_{1}$ is well-defined.
%
Let us prove (\ref{eq:ref4}), 
we denote for all $q_1\in Q_1$,
%\begin{equation}
%\label{eq:ref5}
$\tilde J(q_1) = \min_{q_2\in R(q_1)} J(T_{2,\St_{2,C_\varepsilon}},C_{\varepsilon}(O_s),C_{\varepsilon}(O_t),q_2)$.
%\end{equation}
If $\tilde J(q_1)=0$, then it means that there exists $q_2\in R(q_1)$ such that $J(T_{2,\St_{2,C_\varepsilon}},C_{\varepsilon}(O_s),C_{\varepsilon}(O_t),q_2)=0$.
This implies that $H_2(q_2)\in C_{\varepsilon}(O_t)$. Since $(q_1,q_2) \in R$, Definition~\ref{def:bisim} gives that $d(H_1(q_1),H_2(q_2)) \le \varepsilon$. Hence $H_1(q_1)\in O_t$ and $J(T_{1,{{\St}_1}},O_s,O_t,q_1)=0$. Thus, if $\tilde J(q_1)=0$,~(\ref{eq:ref4})  holds.
We now proceed by induction, let us assume that there exists $k\in \N$ such that for all $q_1 \in Q_1$ such that 
$\tilde  J(q_1)\le k$, (\ref{eq:ref4})  holds. Let $q_1\in Q_1$ such that 
$\tilde J(q_1)=k+1$, let $q_2\in Q_2$ be given by
$
q_2=\argmin_{p_2\in R(q_1)} J(T_{2,\St_{2,C_\varepsilon}},C_{\varepsilon}(O_s),C_{\varepsilon}(O_t),p_2),
$
then, $J(T_{2,\St_{2,C_\varepsilon}},C_{\varepsilon}(O_s),C_{\varepsilon}(O_t),q_2)=k+1$.
This implies $H_2(q_2)\in C_{\varepsilon}(O_s)\setminus C_{\varepsilon}(O_t)$. Since $(q_1,q_2) \in R$,  $d(H_1(q_1),H_2(q_2)) \le \varepsilon$ and $H_1(q_1)\in O_s$. If $H_1(q_1)\in O_t$ then 
$J(T_{1,{{\St}_1}},O_s,O_t,q_1)=0$ and the induction step is completed. 
Hence, let us assume $H_1(q_1)\notin O_t$. 
By~(\ref{eq:ref3}), ${\St}_1(q_1)=\St_{2,C_\varepsilon}(q_2)$. 
Let $l\in {\St}_1(q_1)$, let $q_1'\in \delta_1(q_1,l)$, since $(q_1,q_2)\in R$, there exists, by Definition~\ref{def:bisim}, $q_2' \in \delta_2(q_2,l)$ such that $(q_1',q_2')\in R$. It follows that
$
\tilde J(q_1')\le J(T_{2,\St_{2,C_\varepsilon}},C_{\varepsilon}(O_s),C_{\varepsilon}(O_t),q_2')
$.
Since $H_2(q_2)\in C_{\varepsilon}(O_s)\setminus C_{\varepsilon}(O_t)$, we have by (\ref{eq:rec1}),
$ 
J(T_{2,\St_{2,C_\varepsilon}},C_{\varepsilon}(O_s),C_{\varepsilon}(O_t),q_2')  \le
J(T_{2,\St_{2,C_\varepsilon}},C_{\varepsilon}(O_s),C_{\varepsilon}(O_t),q_2)-1 = k.
$
Therefore,
$
\tilde J(q_1')\le k.
$
Then, by the induction hypothesis we get that 
$
J(T_{1,{{\St}_1}},O_s,O_t,q_1') \le \tilde J(q_1')\le k.
$
Since this holds for all $l\in {\St}_1(q_1)$ and all $q_1'\in \delta_1(q_1,l)$, 
and since $H_1(q_1) \in O_s\setminus O_t$,
we have by (\ref{eq:rec1}) that
$
J(T_{1,{{\St}_1}},O_s,O_t,q_1) \le k+1
$
which completes the induction step.
Thus, we have proved by induction that for all $q_1 \in Q_1$ such that 
$\tilde  J(q_1)\in \N$, (\ref{eq:ref4})  holds.
If $\tilde J(q_1) = +\infty$, then~(\ref{eq:ref4}) clearly holds as well.
\end{proof}

The previous theorem gives us a way by equation~(\ref{eq:ref3}) to concretize a controller for abstraction $T_2$ into a controller for $T_1$. Equation (\ref{eq:ref4}) provides guarantees on the performance of this controller. Particularly, let us remark that
the states of
$T_{1,\St_1}$ from which the control objective is achieved (i.e. the states with finite entry-time) are those
related through the approximate bisimulation relation $R$ to states of $T_{2,\St_{2,C_\varepsilon}}$ with finite entry-time. 
In addition,
if $\St_{2,C_\varepsilon}$ is the time-optimal controller for $T_2$ and reachability specification $(C_\varepsilon(O_s),C_\varepsilon(O_t))$, the following result gives estimates of the distance to optimality for the controller
$\St_1$.  

\begin{theorem}
\label{th:approxreach1} 
Let $\St_{2,C_\varepsilon}^*$, $\St_{2,E_\varepsilon}^*$ be time-optimal controllers for $T_2$ and specification $(C_\varepsilon(O_s),C_\varepsilon(O_t))$ and $(E_\varepsilon(O_s),E_\varepsilon(O_t))$ respectively.
Let $\St_{1}$ be the controller for $T_1$ obtained  from $\St_{2,C_\varepsilon}^*$
by the concretization equation (\ref{eq:ref3}). Let $\St_{1,E_{2\varepsilon}}$ be the controller for $T_1$ obtained 
from $\St_{2,E_{\varepsilon}}^*$ by
\begin{equation}
\label{eq:ref7}
\forall q_1\in Q_1,\; S_{1,E_{2\varepsilon}}(q_1) = \St_{2,E_\varepsilon}^* \left(\argmin_{q_2\in R(q_1)} J^*(T_{2},E_{\varepsilon}(O_s),E_{\varepsilon}(O_t),q_2)\right).
\end{equation}
Then, for all $q_1\in Q_1$,
\begin{eqnarray*}
J^*(T_{1} ,E_{2\varepsilon}(O_s),E_{2\varepsilon}(O_t),q_1) \le &
J(T_{1,{\St_{1,E_{2\varepsilon}}}} ,E_{2\varepsilon}(O_s),E_{2\varepsilon}(O_t),q_1) & \le \\
& J^*(T_{1},O_s,O_t,q_1) & \le \\
& J(T_{1,{{\St}_1}},O_s,O_t,q_1) & \le J^*(T_{1},C_{2\varepsilon}(O_s),C_{2\varepsilon}(O_t),q_1). 
\end{eqnarray*}
\end{theorem}

\begin{proof} The proof essentially follows the same line as the proof of Theorem~\ref{th:approxsafe1}.
The first and third inequalities are direct consequences of the definition of time-optimal value function. 
Let us prove the fourth inequality. 
Let $\St_{1,C_{2\varepsilon}}^*$ be the time-optimal controller for $T_1$ and reachability specification $(C_{2\varepsilon}(O_s),C_{2\varepsilon}(O_t))$. 
From Lemma~\ref{lem:contract}, we have   
$C_{2\varepsilon}(O_s) \subseteq C_{\varepsilon}(C_\varepsilon(O_s))$ and 
$C_{2\varepsilon}(O_t) \subseteq C_{\varepsilon}(C_\varepsilon(O_t))$.
Then, for all $q_1\in Q_1$,
$J^*(T_1,C_{2\varepsilon}(O_s),C_{2\varepsilon}(O_t),q_1) 
 \ge 
J(T_{1,\St^*_{1,C_{2\varepsilon}}},C_{\varepsilon}(C_\varepsilon(O_s)),C_{\varepsilon}(C_\varepsilon(O_t)),q_1)$.
%\begin{eqnarray}
%\nonumber
%J^*(T_1,C_{2\varepsilon}(O_s),C_{2\varepsilon}(O_t),q_1) & = & J(T_{1,\St^*_{1,C_{2\varepsilon}}}, C_{2\varepsilon}(O_s),C_{2\varepsilon}(O_t),q_1)  \\
%\label{eq:1}
%& \ge &
%J(T_{1,\St^*_{1,C_{2\varepsilon}}},C_{\varepsilon}(C_\varepsilon(O_s)),C_{\varepsilon}(C_\varepsilon(O_t)),q_1).
%\end{eqnarray}
%
Now let us define the controller for $T_2$,
$\St_{2,C_\varepsilon}$, 
given, for all $q_2\in Q_2$, by 
$$
\St_{2,C_\varepsilon}(q_2) = 
\St^*_{1,C_{2\varepsilon}}\left(\argmin_{q_1\in R^{-1}(q_2)} J(T_{1,\St^*_{1,C_{2\varepsilon}}},C_{\varepsilon}(C_\varepsilon(O_s)),C_{\varepsilon}(C_\varepsilon(O_t)),q_1)
\right).
$$
By reversing the role of $T_1$ and $T_2$ in Theorem~\ref{th:stat2}, it follows that for all $q_2\in Q_2$,
\begin{eqnarray*}
J(T_{2,{{\St}_{2,C_\varepsilon}}},C_\varepsilon(O_s),C_\varepsilon(O_t),q_2)  &\le& 
 \min_{q_1\in R^{-1}(q_2)} J(T_{1,\St^*_{1,C_{2\varepsilon}}},C_{\varepsilon}(C_\varepsilon(O_s)),C_{\varepsilon}(C_\varepsilon(O_t)),q_1) \\
&\le&
\min_{q_1\in R^{-1}(q_2)}
J^*(T_1,C_{2\varepsilon}(O_s),C_{2\varepsilon}(O_t),q_1).
\end{eqnarray*}
%Then, from (\ref{eq:1}), we have for all $q_2\in Q_2$,
%\begin{equation}
%\label{eq:2}
%J(T_{2,{{\St}_{2,C_\varepsilon}}},C_\varepsilon(O_s),C_\varepsilon(O_t),q_2)  \le 
%\min_{q_1\in R^{-1}(q_2)}
%J^*(T_1,C_{2\varepsilon}(O_s),C_{2\varepsilon}(O_t),q_1).
%\end{equation}
%
Moreover, for all $q_2\in Q_2$,
$
J(T_{2,{{\St}^*_{2,C_\varepsilon}}},C_\varepsilon(O_s),C_\varepsilon(O_t),q_2) \le 
J(T_{2,{{\St}_{2,C_\varepsilon}}},C_\varepsilon(O_s),C_\varepsilon(O_t),q_2)
$
which gives together with Theorem~\ref{th:stat2}, for all $q_1 \in Q_1$,
\begin{eqnarray*}
J(T_{1,{{\St}_1}},O_s,O_t,q_1) &\le &
\min_{q_2\in R(q_1)} J(T_{2,{{\St}_{2,C_\varepsilon}}},C_\varepsilon(O_s),C_\varepsilon(O_t),q_2) \\
& \le & \min_{q_2\in R(q_1)} \min_{p_1\in R^{-1}(q_2) } J^*(T_1,C_{2\varepsilon}(O_s),C_{2\varepsilon}(O_t),p_1) 
 \le  J^*(T_1,C_{2\varepsilon}(O_s),C_{2\varepsilon}(O_t),q_1).
\end{eqnarray*}
%
%Now, it follows from (\ref{eq:2}) that for all $q_1 \in Q_1$,
%$$
%J(T_{1,{{\St}_1}},O_s,O_t,q_1)  \le  
%\min_{q_2\in R(q_1)} \min_{p_1\in R^{-1}(q_2) } J^*(T_1,C_{2\varepsilon}(O_s),C_{2\varepsilon}(O_t),p_1).
%$$
%The second inequality of the theorem is proved by remarking that for all $q_1 \in Q_1$,
%$$
%\min_{q_2\in R(q_1)} \min_{p_1\in R^{-1}(q_2) } J^*(T_1,C_{2\varepsilon}(O_s),C_{2\varepsilon}(O_t),p_1) \le 
%J^*(T_1,C_{2\varepsilon}(O_s),C_{2\varepsilon}(O_t),q_1).
%$$
Let us now prove the second inequality.
From Lemma~\ref{lem:contract},    
$O_s \subseteq C_{\varepsilon}(E_\varepsilon(O_s))$ and 
$O_t \subseteq C_{\varepsilon}(E_\varepsilon(O_t))$.
Then, for all $q_1\in Q_1$,
$
J^*(T_1,O_s,O_t,q_1)  
\ge 
J(T_{1,\St^*_{1}}, C_{\varepsilon}(E_\varepsilon(O_s)),C_{\varepsilon}(E_\varepsilon(O_t)),q_1)$.
Now let us define $\St_{2,E_\varepsilon}$, the controller for $T_2$
given for all $q_2\in Q_2$ by 
$$
\St_{2,E_\varepsilon}(q_2) = {\St}_{1}^*\left(\argmin_{q_1\in R^{-1}(q_2)} J(T_{1,\St^*_{1}}, C_{\varepsilon}(E_\varepsilon(O_s)),C_{\varepsilon}(E_\varepsilon(O_t)),q_1) 
 \right).
$$
By reversing the role of $T_1$ and $T_2$ in Theorem~\ref{th:stat2}, it follows that for all $q_2\in Q_2$,
\begin{eqnarray*}
J(T_{2,\St_{2,E_\varepsilon}}, E_\varepsilon(O_s), E_\varepsilon(O_t),q_2)
 & \le &
\min_{q_1\in R^{-1}(q_2)} J(T_{1,\St^*_{1}}, C_{\varepsilon}(E_\varepsilon(O_s)),C_{\varepsilon}(E_\varepsilon(O_t)),q_1) \\
&\le &
\min_{q_1\in R^{-1}(q_2)}
J^*(T_1,O_s,O_t,q_1).
\end{eqnarray*}
Using the fact that from Lemma~\ref{lem:contract}, $E_\varepsilon(O_s) \subseteq C_{\varepsilon}(E_{2\varepsilon}(O_s))$ and $E_\varepsilon(O_t) \subseteq C_{\varepsilon}(E_{2\varepsilon}(O_t))$, 
we can show similar to Theorem~\ref{th:stat2} that for all $q_1\in Q_1$
\begin{eqnarray*}
J( T_{1,{\St_{1,E_{2\varepsilon}}}} ,E_{2\varepsilon}(O_s),E_{2\varepsilon}(O_t),q_1) & \le &
\min_{q_2\in R(q_1)} J^*(T_{2},E_{\varepsilon}(O_s),E_{\varepsilon}(O_t),q_2) \\
 & \le &
\min_{q_2\in R(q_1)} J(T_{2,\St_{2,E_\varepsilon}}, E_\varepsilon(O_s), E_\varepsilon(O_t),q_2) \\
& \le & \min_{q_2\in R(q_1)} \min_{p_1\in R^{-1}(q_2)}
J^*(T_1,O_s,O_t,p_1) \le J^*(T_1,O_s,O_t,q_1).
\end{eqnarray*}
\end{proof}

By computing the controllers $\St_1$ and $\St_{1,E_{2\varepsilon}}$ one is able to give a certified upper-bound on the distance to optimality of the controller $\St_1$ we will 
use to control $T_1$.
Moreover, if the reachability problem is robust (i.e. the time-optimal value function depends continuously on the specification); then 
$J^*(T_1,O_s,O_t,q_1)$ can be approximated arbitrarily close.
We would like to highlight the differences of our approach with~\cite{mazo2010}
where time-optimal reachability controllers are synthesized using discrete abstractions related by
alternating simulations. These are weaker assumptions than those considered in the present work.
In~\cite{mazo2010}, an approach to compute guaranteed upper and lower bounds of the value function
is given. However, contrarily to our approach there is no clue that these lower and upper bounds 
can approach the true time-optimal value function arbitrarily close. Also, in~\cite{mazo2010}, the 
controllers are refined via the natural concretization procedure that suffers from the drawbacks described in 
Section~\ref{sec:prel} whereas our approach does not.

\section{Application to Switching Controller Design}

In this section, we present an effective approach to switching controller design based on
the  synthesis approaches introduced in this paper in combination
with the approximately bisimilar discrete abstractions of switched systems developed in~\cite{B-GirPolTab08}.
%After recalling briefly the main results of~\cite{B-GirPolTab08}, we give an illustration by 
%applying our approach to synthesize controllers for a DC-DC converter.

%\subsection{Approximately Bisimilar Abstractions for Switched Systems}

%To motivate the development of controller synthesis techniques using approximately bisimilar abstractions, 
%We recall briefly the main results of~\cite{B-GirPolTab08}.

\begin{definition}
\label{def:switch} A {\it switched system} is a triple
$
\Sigma=(\R^n,P,F),
$
where $\R^n$ is the state space;
$P =\{1,\dots,m\}$ is the finite set of modes;
$F=\{f_{1},\dots,f_{m}\}$ is a collection of vector
fields indexed by $P$. 
\end{definition}

%A {\it switching signal} of $\Sigma$ is a piecewise constant function $\mathbf{p}: \R^+ \rightarrow P$, continuous from the right and with a finite number of discontinuities on every bounded interval of $\R^+$.
%The discontinuities of $\mathbf{p}$
%are called {\it switching times}.
%A piecewise $\mathcal C^1$ function $\mathbf{x}:\R^+ \rightarrow \R^n$ is
%said to be a {\it trajectory} of $\Sigma$ if it is continuous and there exists a switching
%signal $\mathbf p$ such that, at each
%$t\in \R^+$ where the function $\mathbf{p}$ is continuous, $\mathbf{x}$ is continuously differentiable and satisfies
%$
%\dot{\mathbf{x}} (t) = f_{\mathbf{p}(t)}(\mathbf{x}(t)).
%$
%We denote $\mathbf{x}(t,x,\mathbf{p})$ the point
%reached at time $t\in \R^+$ from the initial condition $x$ under
%the switching signal $\mathbf{p}$. 
%We denote $\mathbf{x}(t,x,{p})$ the point reached by $\Sigma$ at time $t\in \R^+$ from the initial condition $x$ under the constant switching signal $\mathbf{p}(t)=p$, for all $t\in \R^+$.

Given a switched system $\Sigma=(\R^n,P,\mathcal P, F)$ and a parameter $\tau >0$,
we define a transition system $T_\tau(\Sigma)$ that describes trajectories of duration $\tau$ of $\Sigma$.
This can be seen as a time sampling process.
This is natural when the switching in $\Sigma$ is determined by a time-triggered controller with period $\tau$.
Formally, 
$T_{\tau}(\Sigma)=(Q_1,L,\delta_1,O,H_1)$
where the set of states is $Q_1=\R^n$; the set of actions is $L=P$;
the transition relation is given by 
$x'\in \delta_1(x,p)$ if and only if the solution of the differential equation $\dot{\mathbf{x}}(t)=f_p(\mathbf{x}(t))$ with 
$\mathbf{x}(0)=x$ satisfies $\mathbf{x}(\tau)=x'$;
the set of outputs is $O=\R^n$;
the observation map $H_1$ is the identity map over $\R^n$.
The set of
observations $O=\R^n$ is equipped with the metric $d(x,x')=\|x-x'\|$ where $\|.\|$ is the usual Euclidean norm.
The existence of approximately bisimilar discrete abstractions of $T_\tau(\Sigma)$ is related to the notion of incremental stability~\cite{angeli2002}. 
%Intuitively, incremental stability of a switched system means that all the trajectories associated with the same switching signal converge to the same reference trajectory independently of their initial condition. 
Under this assumption, it is possible to compute approximately bisimilar discrete abstractions of arbitrary precision for $T_\tau(\Sigma)$ based on a griding of the state-space. Moreover, the approximate bisimulation relation can be fully characterized by a Lyapunov function proving incremental stability of switched system $\Sigma$ (see~\cite{B-GirPolTab08}
for details).

Controller synthesis for switched systems with safety or reachability specifications can be tackled
by direct application of fixed-point computation or dynamic programming using guaranteed over-approximations~\cite{Asarin2000}
or convergent approximations~\cite{Mitchell2000} of reachable sets. In the first case, the synthesized controllers are ``correct by design'' but there is no guarantee that the synthesis algorithm will terminate. In the second case, we can only prove
that the synthesized controllers are ``correct in the limit'' in the sense that correct controllers can be approximated
arbitrarily close. The approach described in this paper does not have these problems but only applies to incrementally
stable systems.
For illustration purpose, we apply our approach to  
a boost DC-DC converter. 
It is a switched system with modes,
the two dimensional dynamics associated with both modes are affine of the form
$
\dot x(t)=A_p x(t) +b$ for $p=1,2$ (see~\cite{B-GirPolTab08} for numerical values).
%with 
%$$
%A_1=\left[ 
%\begin{smallmatrix}
%-\frac{r_l}{x_l} & 0 \\
%0  & -\frac{1}{x_c}\frac{1}{r_0+r_c}
%\end{smallmatrix}
%\right],
%\;
%A_2=\left[ 
%\begin{smallmatrix}
%-\frac{1}{x_l}(r_l+\frac{r_0r_c}{r_0+r_c}) & -\frac{1}{x_l}\frac{r_0}{r_0+r_c} \\
%\frac{1}{x_c}\frac{r_0}{r_0+r_c}  & -\frac{1}{x_c}\frac{1}{r_0+r_c}
%\end{smallmatrix}
%\right],\;
%b=\left[ 
%\begin{smallmatrix}
%\frac{v_s}{x_l} \\
%0
%\end{smallmatrix}
%\right].
%$$
%We use the numerical values from~\cite{beccuti2005}; for a better numerical conditioning, we rescaled the second 
%variable of the system (see~\cite{B-GirPolTab08}).
It can be shown that it is incrementally stable and thus approximately bisimilar discrete abstractions
can be computed.

We first consider the problem of regulating the state of the DC-DC converter around a desired nominal state. 
This can be done for instance by synthesizing a controller that keeps the state of the switched system in a set centered around the nominal state.
This is a safety specification. In the following, we consider the specification given by the set $O_s= [1.1,1.6] \times [5.4,5.9]$. We use a time sampling parameter $\tau=1$ and choose to work with a discrete abstraction that is approximately bisimilar to $T_\tau(\Sigma)$ with precision $\varepsilon=0.05$.
We compute a safety controller for the switched system $T_\tau(\Sigma)$ by the approach described in Section~\ref{sec:safe}.
%We first have to compute the maximal safety controller $\St_{2,C_\varepsilon}^*$ for the abstraction and the specification given by the contracted
%safe set
%$
%C_\varepsilon(O_s)=[1.15,1.55]\times[5.45,5.85].
%$
The discrete abstraction has a finite number of states inside $H_2^{-1}(C_\varepsilon(O_s))$
(actually $169744$).
The fixed point algorithm for the synthesis of the maximal safety controller for the abstraction and specification $C_\varepsilon(O_s)$ terminates in $2$ iterations. 
The resulting safety controller $\St_1$ for the switched system  $T_\tau(\Sigma)$ and the specification $O_s$
is shown on the left part of Figure~\ref{fig:ex1b} where we have represented a trajectory of the system where the switching is controlled using a lazy implementation of the controller $\St_1$:
when the controller has the choice between mode 1 and 2, it keeps the current mode active. We can check that the specification is effectively met.
We also compute the upper-bound of the maximal safety controller $\St_1^*$ for switched system $T_\tau(\Sigma)$ and specification $O_s$, given by Theorem~\ref{th:approxsafe1}.
The abstraction has $383161$ states inside $H_2^{-1}(E_\varepsilon(O_s))$ and
the fixed point algorithm for computing the maximal safety controller terminates in $4$ iterations.
%This controller is shown on the bottom-left part of Figure~\ref{fig:ex1a}.
The resulting controller $\St_{1,E_{2\varepsilon}}$, shown on the right side  of Figure~\ref{fig:ex1b},
is an upper-bound of the maximal safety controller $\St_1^*$ for $T_\tau(\Sigma)$ and specification $O_s$.

\begin{figure}[!t]
%\vspace{-3cm}
\begin{center}
%\hspace{-0.6cm}
\includegraphics[angle=0,scale=0.5]{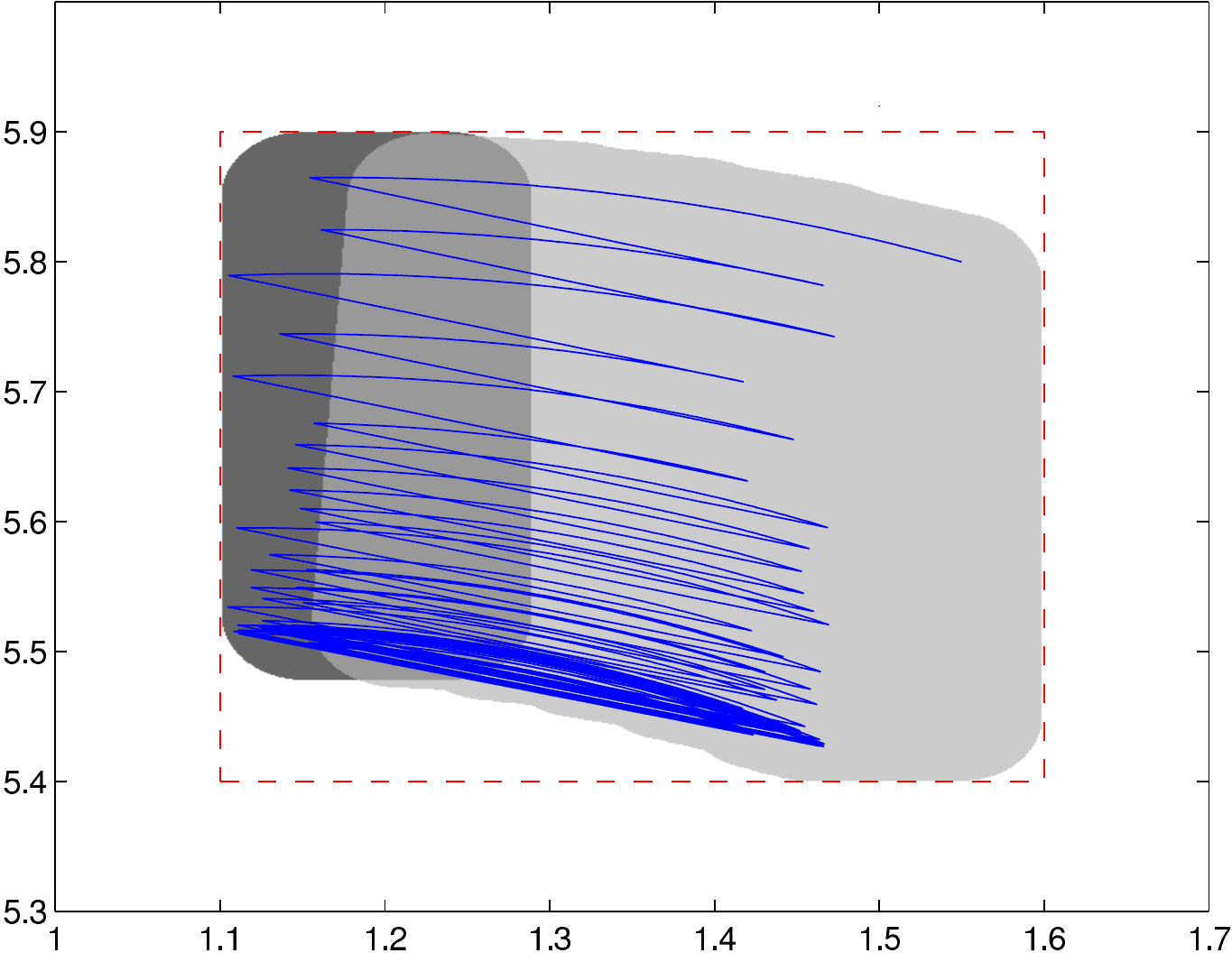}
%\hspace{-0.5cm}
\includegraphics[angle=0,scale=0.5]{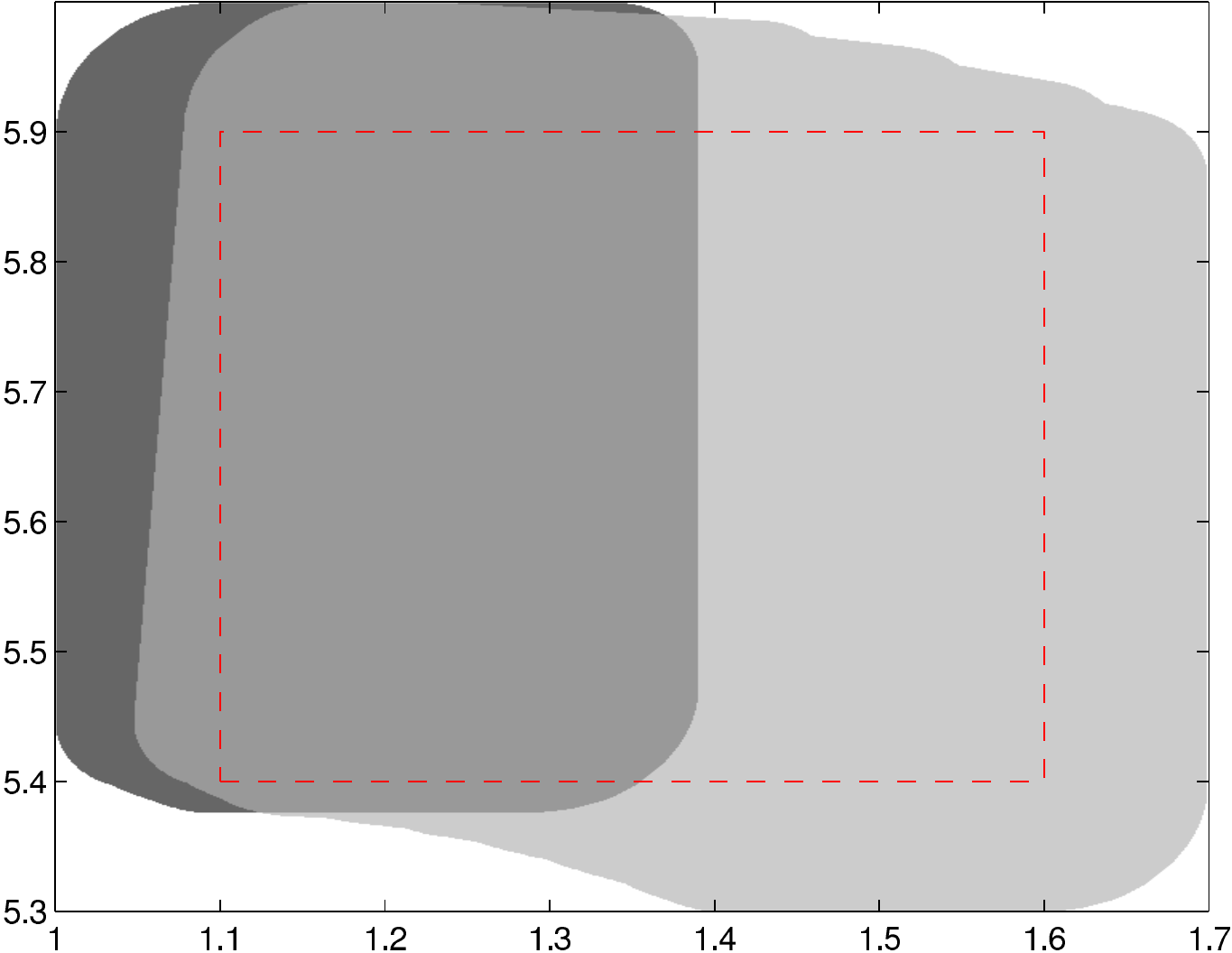}
%\vspace{-0.4cm}
\caption{Safety controller $\St_1$ for the switched system $T_\tau(\Sigma)$ and specification $O_s$ with controlled trajectory (left);
Safety controller $\St_{1,E_{2\varepsilon}}$ for the switched system $T_\tau(\Sigma)$ and specification $E_{2\varepsilon}(O_s)$ (right); dark gray: mode 1, light gray: mode 2, medium gray: both modes are acceptable, white: no action is allowed.
The maximal safety controller $\St_1^*$ for $T_\tau(\Sigma)$ and specification $O_s$ satisfies $\St_1\preceq \St_1^* \preceq \St_{1,E_{2\varepsilon}}$.
}
\label{fig:ex1b}
\end{center}
\end{figure}

We now consider the problem of steering in minimal time the state of the DC-DC converter
in the desired region of operation while respecting some safety constraints. This is a time-optimal control problem.
We consider the specification given by the safe set $O_s= [0.65,1.65] \times [4.95,5.95]$ and the target set 
$O_t = [1.1,1.6] \times [5.4,5.9]$. 
This time, we use a time sampling parameter $\tau=0.5$ and
 choose to work with a discrete abstraction that is approximately bisimilar to $T_\tau(\Sigma)$ with precision $\varepsilon=0.1$.
We compute a suboptimal reachability controller for the switched system $T_\tau(\Sigma)$ by the approach described in Section~\ref{sec:reach}.
The discrete abstraction  has a finite number of states inside $H_2^{-1}(C_\varepsilon(O_s))$
(actually $674041$).
The dynamic programming algorithm for the synthesis of the time-optimal controller for 
the abstraction and reachability specification $(C_\varepsilon(O_s),C_\varepsilon(O_t))$ terminates in $94$ iterations.
The resulting suboptimal controller $\St_1$ for the switched system $T_\tau(\Sigma)$ for the reachability specification $(O_s,O_t)$
is shown on Figure~\ref{fig:ex2b} where we have also represented trajectories of the system where the switching is controlled using the synthesized
controller. 
We can check that the specification is effectively met. 
The entry time associated to $\St_1$, $J(T_{\tau}(\Sigma)_{\St_{1}},O_s,O_t,q_1)$ 
shown on the left part of Figure~\ref{fig:ex2c}, gives an upper-bound of the time-optimal value function.
We also compute the lower-bound of the time-optimal value function
% $J^*(T_{\tau}(\Sigma),O_s,O_t,q_1)$ 
%for switched system $T_\tau(\Sigma)$ and reachability specification $(O_s,O_t)$, 
given by Theorem~\ref{th:approxreach1}.
The abstraction has $1520289$ states inside $H_2^{-1}(E_\varepsilon(O_s))$ and
the fixed point algorithm for computing the maximal safety controller terminates in $66$ iterations.
%This controller is shown on the bottom-left part of Figure~\ref{fig:ex1a}.
The resulting controller $\St_{1,E_{2\varepsilon}}$ with entry time 
$J(T_\tau(\Sigma)_{\St_{1,E_{2\varepsilon}}},E_{2\varepsilon}(O_s),E_{2\varepsilon}(O_t),q_1)$
provides a lower-bound of the time-optimal value function $J^*(T_\tau(\Sigma),O_s,O_t,q_1)$.
This lower-bound is shown on the right side of Figure~\ref{fig:ex2c}.

\begin{figure}[!t]
%\vspace{-3cm}
\begin{center}
\includegraphics[angle=0,scale=0.5]{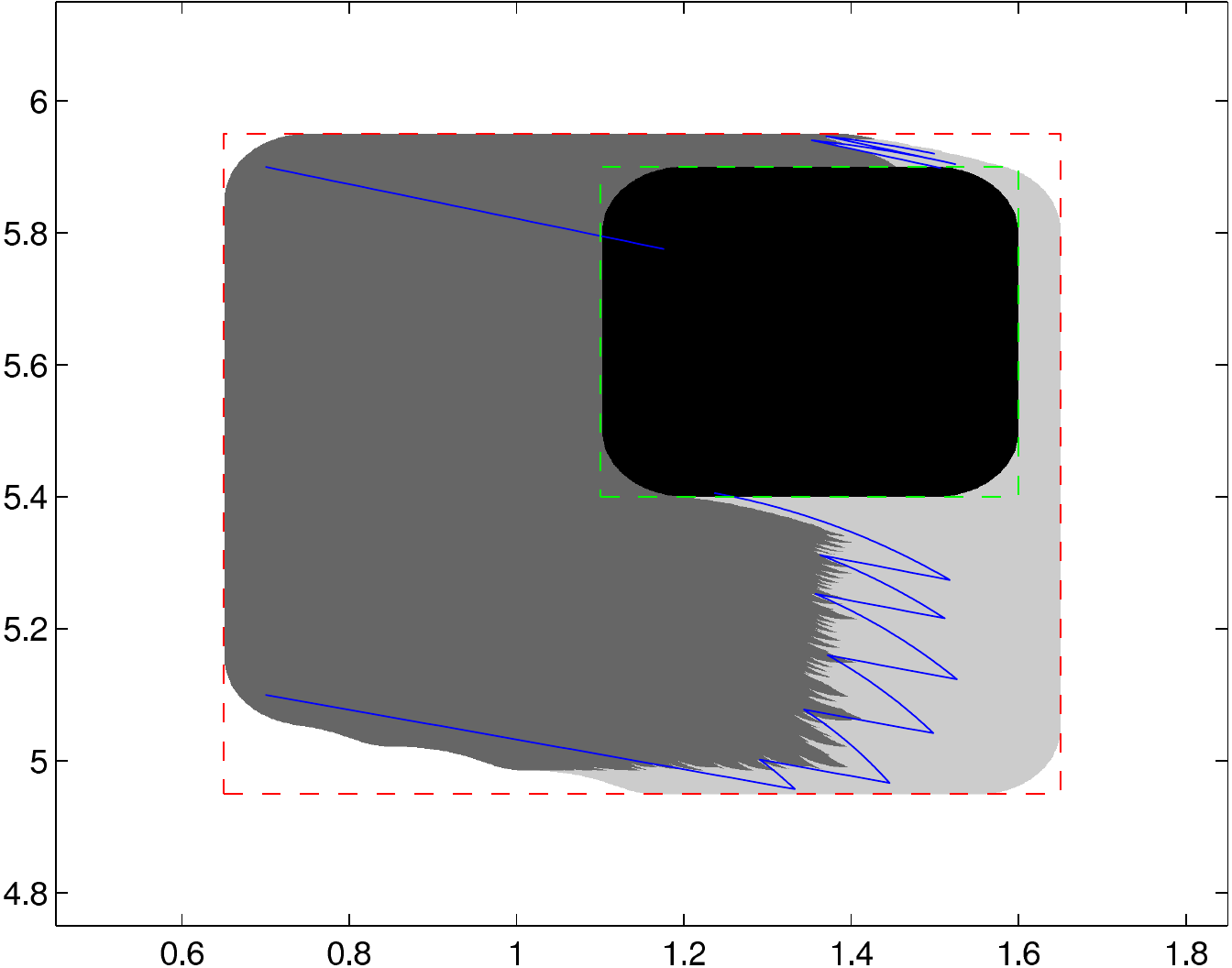}
%\vspace{-0.4cm}
\caption{Suboptimal controller $\St_1$ for the switched system $T_\tau(\Sigma)$ and reachability specification $(O_s,O_t)$  
and trajectories of the controlled switched system.}
\label{fig:ex2b}
\end{center}
\end{figure}

\begin{figure}[!t]
%\vspace{-3cm}
\begin{center}
%\hspace{-0.6cm}
\includegraphics[angle=0,scale=0.5]{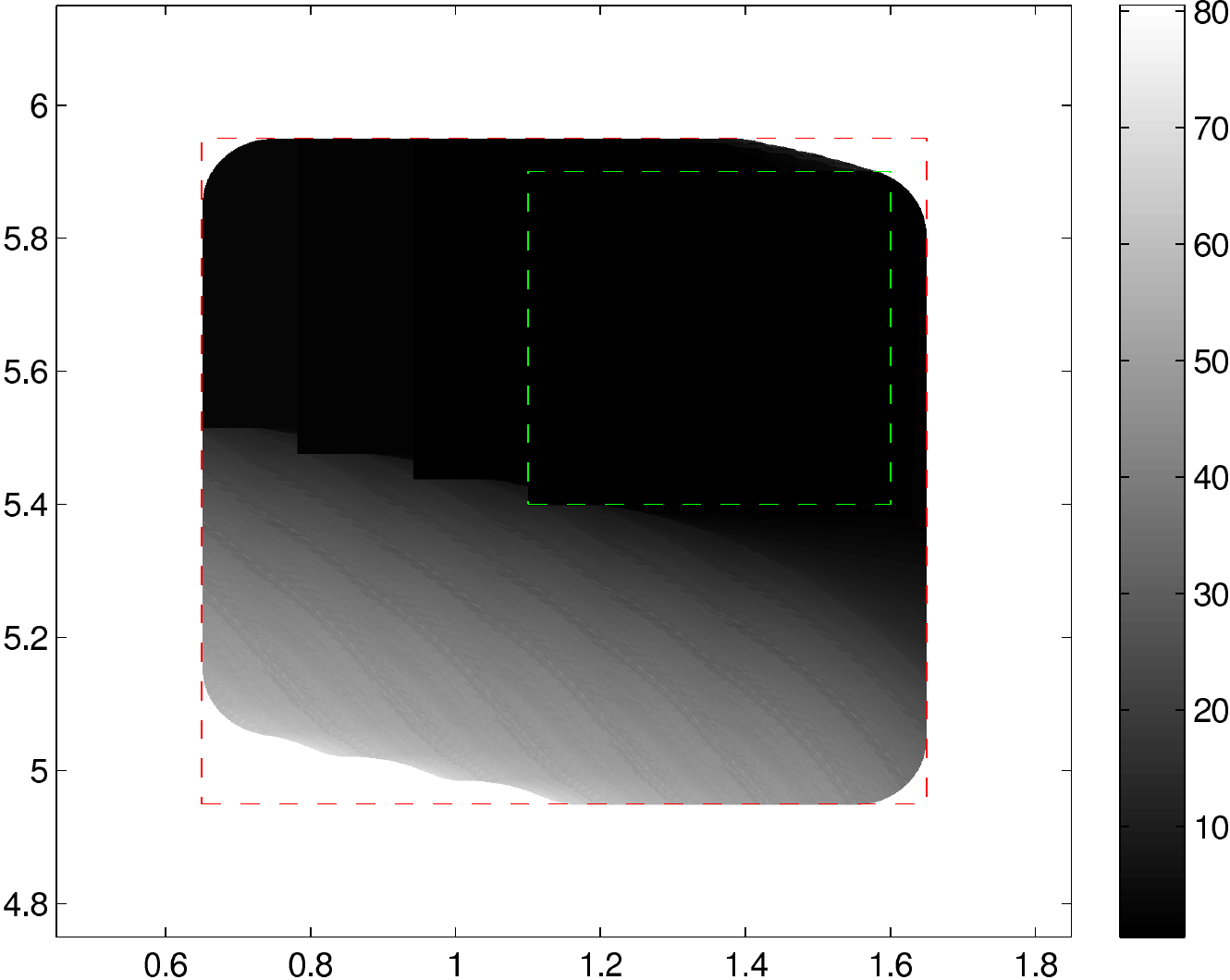}
%\hspace{-0.5cm}
\includegraphics[angle=0,scale=0.5]{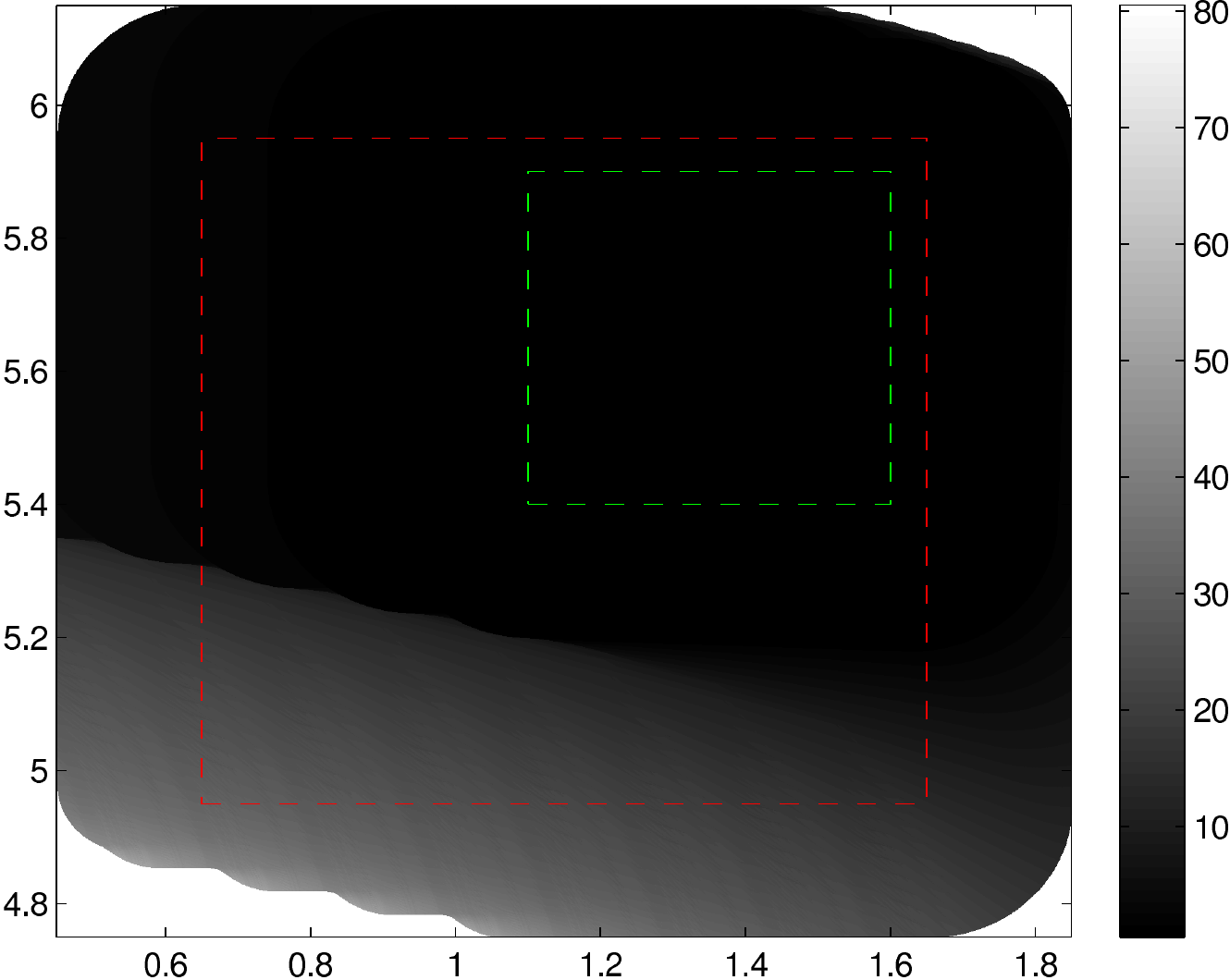}
%\vspace{-0.4cm}
\caption{Entry-time $J(T_\tau(\Sigma)_{\St_{1}},O_s,O_t,q_1)$ for the controller $\St_1$ shown in Figure~\ref{fig:ex2b} (left);
Entry-time $(T_\tau(\Sigma)_{\St_{1,E_{2\varepsilon}}},E_{2\varepsilon}(O_s),E_{2\varepsilon}(O_t),q_1)$ (right);
The time-optimal value function for the switched system $T_\tau(\Sigma)$ and reachability specification $(O_s,O_t)$ satisfies  $J(T_\tau(\Sigma)_{\St_{1,E_{2\varepsilon}}},E_{2\varepsilon}(O_s),E_{2\varepsilon}(O_t),q_1)\le J^*(T_\tau(\Sigma),O_s,O_t,q_1) \le
J(T_\tau(\Sigma)_{\St_{1}},O_s,O_t,q_1)$.}
\label{fig:ex2c}
\end{center}
\end{figure}

\section{Conclusion}

In this paper, we proposed a methodology, based on the use of approximately bisimilar discrete abstractions, for effective computation of controllers for safety and reachability specifications. 
We provided guarantees of performances of the resulting controllers by giving estimates of the distance of the synthesized controller to the maximal (i.e the most permissive) safety controller or to the time-optimal reachability controller. We showed the effectiveness of our approach by synthesizing controllers for a switched system. Let us remark that the techniques presented in the paper are independent of the type of abstractions considered as long as these are approximately bisimilar. 
Future work will deal with the development of similar approaches to handle different optimal control problems and richer specifications given e.g. in temporal logic.

%%%%%%%%%%%%%%%%%%%%%%%%%%%%%%%%%%%%%%%%%%%%%%%%%%%%%%%%%%

\subsubsection*{Acknowledgments} The author would like to thank Gunther Rei{\ss}ig
for his valuable comments on an earlier version of this paper.

\bibliographystyle{alpha}
\bibliography{symcon}

\end{document}